\documentclass[12pt]{article}
\usepackage{footnote}
\makesavenoteenv{tabular}
\makesavenoteenv{table}

\usepackage[pass]{geometry}
\usepackage{tikz}
\usetikzlibrary{patterns}
\usetikzlibrary{arrows,automata}
\usepackage{graphicx}
\usepackage{amsmath}
\usepackage{amssymb}
\usepackage{amsthm}
\newtheorem{theorem}{Theorem}[section]
\newtheorem{lemma}[theorem]{Lemma}

\newtheorem{obs}[theorem]{Observation}

\title{Small-World Formation via Local Information}

\author{Soroush Alamdari \\ {{\small \tt sorush@gmail.com}}}

\date{}

\begin{document}
\maketitle
\abstract{
It has been observed that almost anyone is acquainted with almost anyone else through only a few intermediary links. This has been known as the small-world phenomenon. In this script we investigate this observation from a theoretical stand-point by imagining each individual as a greedy agent. We show that small-world properties emerge naturally when individuals pay for links to manage a pull from the population.\footnote{The results in this paper were initially announced in the 35th ACM SIGACT-SIGOPS Symposium on Principles of Distributed Computing (PODC 2016).}

\section{The small-world phenomenon}
The drives that move us towards each other are far more complicated than to be mathematically understandable, and yet, there could be clues that may help us paint an abstract picture. For this, we focus on one particular observation, that is, the small-world phenomenon: The idea that any two people know each other through only a few links.

The small-world phenomenon was confirmed in an experiment in 1960s by Milgram \cite{milgram67smallworld},
who sent randomly selected individuals in Nebraska and Kansas a packet with instructions asking them to help the packet reach a certain businessman in Boston. In case they did not know the target individual personally, the recipients were instructed to pass the packet on to someone they knew personally who was most likely to personally know the target. So each initial recipient started a chain of correspondence, some of which ended before reaching the target due to lack of participation. Among the chains that reached the target, the average length of a chain was just below six. An observation we have been trying to understand since.

This is a curious phenomenon to the mathematician since it satisfies two opposing criteria:
\begin{enumerate}
\setcounter{enumi}{0}
\item \label{const:degree} \textbf{Degree:} The number of others that any node has links to, that is the \emph{degree} of that node, is small relative to the total population.
\item \label{const:chain} \textbf{Diameter:} 
    For any recipient and any target, the length of the shortest chain of links that starts from the recipient and ends in the target is very small. 
    The length of the longest such chain in a network is referred to as the \emph{diameter} of that network.
\end{enumerate}

In 1990s Watts and Strogatz \cite{watts1998cds} showed that in a society with a population of $n$ that are sitting around a ring, if each individual has links to the $n^\varepsilon$ others that are sitting on each side of it plus $n^\varepsilon$ random other nodes, then each pair of nodes are connected via a chain of expected length smaller than $1/\varepsilon$. While this simple model satisfies both Constraints \ref{const:degree} and \ref{const:chain}, it does not explain how such a chain is actually found by the forwarded packages in the experiment. There each recipient is unaware of the structure of this network and cannot decide who among the people it has a link to is on the shortest chain to the target. The fact that such short chains are actually observed by the experiment implies a much more restricting constraint than Constraint \ref{const:chain}. This observation allows us to strengthen Constraint \ref{const:chain}:
\begin{enumerate}
\setcounter{enumi}{2}
\item \label{const:private} \textbf{Limited information:} Each node is unaware of the links maintained by other nodes. We refer to this condition as the limited information constraint.
\item \label{const:gchain} \textbf{Routing diameter:}
For any recipient and any target, a very short chain connecting the recipient to the target can be found if each recipient chooses who to forward the packet to based on similarity with the target. We refer to the length of the longest such chain in a network as the \emph{routing diameter} of that network.
\end{enumerate}

\begin{table*}[t]
\centering
\begin{tabular}{|c|c|c|c|}
  \hline
   & $\alpha<d+1$ & $\alpha=d+1$ & $d+1<\alpha$\\
  \hline
  Maximum degree & $o( n^{1-\delta\frac{\alpha}{d+1}}\log \log n)$ & $O(\log n)$ & $O(\log_{\frac{\alpha}{d+1}} \log n)$ for $d=1$  \\
  \hline
   Routing diameter & $O(\log_{\frac{d+1}{\alpha}} \log n)$ & $O(\log n)$ & $o(n^{1-\frac{\delta}{\alpha-d}}\log\log n)$ \\
  \hline
\end{tabular}
\caption{The bounds on maximum degree and  routing diameter at stability for different values of $\alpha$ relative to the dimension $d$. Here $\delta$ can be any constant smaller than $1$.}
\label{tab:tradeoff}
\end{table*}

Travers and Milgram \cite{Travers69anexperimental} observed that participants favor geographical information when deciding the next recipient of the packet. Kleinberg \cite{Kleinberg:2000:STOC, Kleinberg:2000:nature} argues that such information is sufficient for satisfying Constraint \ref{const:gchain} by showing that in a population that is distributed homogeneously over the plane, if each node has a link to a node of geographical distance $l$ with probability $1/l^\alpha$, then iteratively choosing the next recipient based on geographical proximity to the target almost always produces a short chain. Particularly, when $\alpha=2$, both the expected length of the chain and the expected number of links that each node creates can be bounded by $O(\log n)$.

It remains curious whether the small-world phenomenon exists by coincidence, or if it is designed by our collective minds. Our goal here is to find strategies that may have produced networks of such properties:
\begin{enumerate}
\setcounter{enumi}{4}
\item \label{const:strategic} \textbf{Strategic choice:} Each node creates links by making choices toward some objective.
\end{enumerate}

Many game-theoretic models of network creation have been analyzed. However, by considering a model where individuals make choices strategically based on the choices of others, our Constraint \ref{const:private} is violated.

For example consider the model by Even-Dar and Kearns \cite{Even-DarK06}, where the population is distributed uniformly on the plane and each agent $v$ pays for links in order to reduce its average distance to
other nodes through the network. In their model a link to a node of geographical distance $l$ induces a cost of $l^\alpha$. Particularly, they show that in this model when $\alpha<2$, the diameter can be bounded by a constant. Calculating the objective function that each node works with requires knowledge of the whole network.

In another related work Clauset and Moore \cite{Clauset2003} show that if nodes rewire their outgoing link whenever they cannot reach another node with bounded number of steps, then small-world properties appear with bounds similar to what was observed by Kleinberg \cite{Kleinberg:2000:STOC, Kleinberg:2000:nature}. They show this via repeated simulations and while their model relaxes full-information assumption, it does not satisfy our Constraint \ref{const:private}.

When attempting to derandomize the work of Kleinberg \cite{Kleinberg:2000:STOC, Kleinberg:2000:nature} at the critical point where logarithmic bounds appear on both expected degree and routing diameter, we come upon a strategic model of network creation that achieves similar bounds and satisfies all of our constraints. We interpret the formula that governs the decision making of agents in the resulting model as a balancing of the cost of maintaining links and a desire for connectivity.

In Section \ref{sec:model} we present the model in detail and discuss the components of the cost function that the agents are trying to minimize.
In Section \ref{sec:analysis} we analyze the model on a $d$-dimensional grid\footnote{Throughout the paper we assume that $d$ is a constant.}.
We focus on a notion of equilibrium where agents cannot further improve their cost by adding or deleting a link, showing that $\alpha=d+1$ is a threshold where the upper bounds on the performance of the greedy routing and maximum degree match to $O(\log n)$.
We also analyze the maximum degree and the performance of greedy routing for other values of $\alpha$, particularly showing that when $\alpha<d+1$
the number of iterations that are required for greedy routing to find its destination can be bounded
by $O(\log_{\frac{d+1}{\alpha}} \log(n))$.

We conclude in Section \ref{sec:conclusion} by discussing the implications of our results.

\section{The Model}
\label{sec:model}
Consider a $d$-dimensional circular grid of $n$ nodes representing the underlying structure upon which the network is formed. For nodes $u$ and $v$ we use $d(u,v)$ to refer to the distance\footnote{In this paper, unless specified otherwise, by \textit{distance} we are referring to the distance through the grid, that is, grid-distance} of $u$ and $v$.
Also, for a node $u$ and a set of nodes $V$ we use $d(u,V)$ to refer to the distance of $u$ to its closest member in $V$, that is, $d(u,V) = \min_{v\in V} d(u,v)$.

A network $N$ assigns to each node $v$ a set $N(v)$ of nodes that $v$ has a links to.
To model the creation of such a network we introduce a cost that each agent incurs at a given setting.
Let $c_N(v)$ be such a function that calculates the cost induced to a node $v$ in $N$. There are two components to this cost function; the cost of maintaining the links that $v$ has to other nodes and the cost induced by what $v$ is missing regarding the nodes that it does not have a link to. For a node $u$, let $l_N(v,u)$ be the former and $s_N(v,u)$ be the latter. We have
$$c_N(v)=\sum_{u\in N(v)} l_N(v,u) + \sum_{u\notin N(v)} s_N(v,u).$$

For the cost that $v$ incurs for maintaining a link to an agent at distance $d(v,u)$, it is natural to assume a dependence between this cost and distance.
Let us define $l_N(v,u) = d^\alpha(v,u)$ for some number $\alpha \in \mathbb{R}$. The exact choice of $\alpha$ is discussed in section \ref{sec:analysis}. We will refer to this cost as the \emph{link-cost}.

It remains to define $s_N(v,u)$, that is, the cost that $v$ pays for a node $u$ that it has no link to.
We capture this by assuming that each node $u$ induces a cost to $v$ that is proportional to the distance between $u$ and the node closest to $u$ that $v$ has a link to,
that is, $d(u,N(v))$. We refer to this as the \emph{separation-cost}.

We are ready to express our cost function $c_N(v)$ such that it only depends on $N(v)$ and the position of the nodes on the underlying grid:
\begin{align}
\label{fml:maincost}
c_N(v) = \beta\sum_{u\in N(v)} d^\alpha(v,u) + \sum_{u\notin N(v)}d(u,N(v)).
\end{align}

\section{Analysis}
\label{sec:analysis}
In this section we analyze the network constructed by agents who seek to reduce their costs by adding and deleting links to other nodes. In particular, we discuss the trade-off between the greedy routing diameter of the network and the maximum degree among nodes.
When $\alpha < \frac{1}{\log_{\beta^{-1}} n}$ every link has a cost smaller than $1$ and therefore the diameter is $1$ and all of the degrees are equal to $n-1$, and when $\alpha > \log(\beta^{-1}n^2)$ no agent can justify a link to a node of distance $>1$, and therefore the degrees and diameter are $2d$ and $dn^{1/d}$ respectively. Here we provide analysis for intermediary values of $\alpha$.

To establish our bounds, we do not require the assumption that each agent chooses their links to minimize their cost. In each of our cases, the assumption that no agent can reduce its cost by adding a link (\emph{add-stability}) is sufficient to prove the upper bound
on the routing diameter of the network. When in addition to add-stability we also assume that no agent can improve its cost by removing a link (\emph{toggle-stability}) we are able to also prove upper bounds on the maximum degree. Since each agent reduces its cost with each of these additions and deletions of links, therefore, performing these operations in any order will eventually result in stability. It is also worth noting that these assumptions are much more general than, for example, the assumption of a Nash equilibrium.

We start with a set of observations on the properties of the grid that will be used throughout this section in order to bound the routing diameter and maximum degree among the nodes. For this, let us define $S(v,l)$ to be the set of all nodes of distance exactly $l$ from $v$, and $B(v,l)$ be the set of all nodes $u$ with $d(u,v)\leq l$. The following observation lists some properties regarding $S(v,l)$ and $B(v,l)$.

\begin{figure}
\begin{tikzpicture}[darkstyle/.style={circle,draw,fill=gray!40,minimum size=20}]
\begin{scope}[shift = {(0.5,0)},scale = 0.9]
	\node [darkstyle]  (v) at (0,8) {$v$};
	\draw (5,5) -- (8,8) -- (0,16);
	\draw (3,5) -- (6,8) -- (0,14);
	\draw [<->,>=stealth',shorten >=1pt, auto, semithick, scale = 1, transform shape] (v)  -- node[above] {$l$}  ++ (6,2);
	\draw [<->,>=stealth',shorten >=1pt, auto, semithick, scale = 1, transform shape] (v)  -- node[above left] {$l-l'$}  ++ (3,3);	
	\draw [blue] (6,8) -- node {} ++ (1,1);
	\draw [blue] (5,9) -- node {} ++ (1,1);
	\draw [blue] (4,10) -- node {} ++ (1,1);
	\draw [blue] (3,11) -- node {} ++ (1,1);
	\draw [blue] (2,12) -- node {} ++ (1,1);
	\draw [blue] (1,13) -- node {} ++ (1,1);	
	\draw [blue] (6,8) -- node {} ++ (1,-1);
	\draw [blue] (5,7) -- node {} ++ (1,-1);
\end{scope}
\begin{scope}[shift = {(10,0)}, scale = 0.9]
	\node [darkstyle]  (u) at (0,8) {$v$};
	\draw (5,5) -- (8,8) -- (0,16);
	\draw (3,5) -- (6,8) -- (0,14);
	\draw [<->,>=stealth',shorten >=1pt, auto, semithick, scale = 1, transform shape] (u)  -- node[above] {$l$}  ++ (6,2);
	\draw [<->,>=stealth',shorten >=1pt, auto, semithick, scale = 1, transform shape] (u)  -- node[above left] {$l-l''$}  ++ (3,3);	
	\draw [blue] (6,8) -- node {} ++ (1,1);
	\draw [blue] (5.5,8.5) -- node {} ++ (1,1);	
	\draw [blue] (5,9) -- node {} ++ (1,1);
	\draw [blue] (4.5,9.5) -- node {} ++ (1,1);
	\draw [blue] (4,10) -- node {} ++ (1,1);
	\draw [blue] (3.5,10.5) -- node {} ++ (1,1);
	\draw [blue] (3,11) -- node {} ++ (1,1);
	\draw [blue] (2.5,11.5) -- node {} ++ (1,1);
	\draw [blue] (2,12) -- node {} ++ (1,1);
	\draw [blue] (1.5,12.5) -- node {} ++ (1,1);	
	\draw [blue] (1,13) -- node {} ++ (1,1);	
	\draw [blue] (6,8) -- node {} ++ (1,-1);
	\draw [blue] (5.5,7.5) -- node {} ++ (1,-1);	
	\draw [blue] (5,7) -- node {} ++ (1,-1);	
	\draw [blue] (0,15) -- node {} ++ (7.5,-7.5);
	\draw [blue] (7.5,8.5) -- node {} ++ (-3.5,-3.5);
\end{scope}
\end{tikzpicture}
\caption{Items \ref{lem:item:DisEqkCover} (left) and \ref{lem:item:DisLeqkCover} (right) of Observation \ref{lem:basiclist}. In each case $l'$ is the diameter of the smaller squares.}
\end{figure}
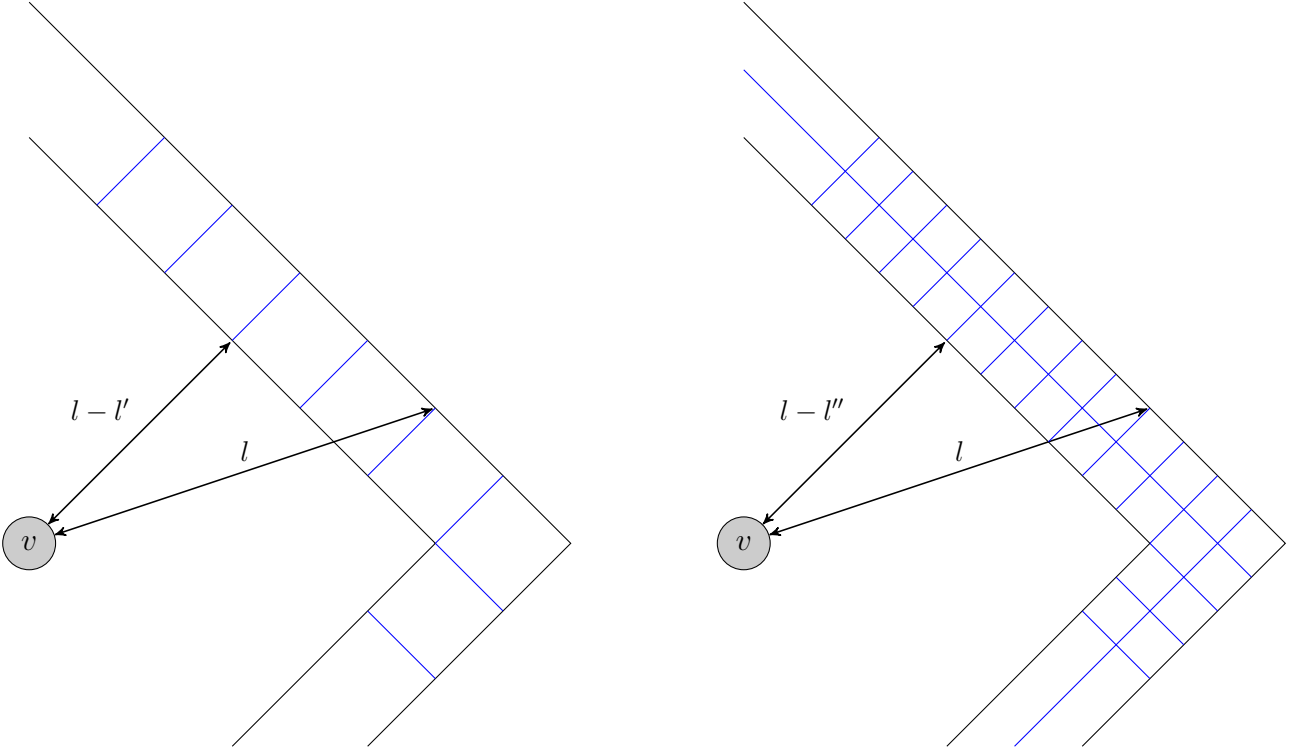

\begin{obs}
\label{lem:basiclist}
For any node $v$ on the $d$-dimensional grid and $l\in \mathbb{N}$ the following holds:
\begin{enumerate}
\item \label{lem:item:DisEqk} $|S(v,l)|\in \Theta(l^{d-1})$.
\item \label{lem:item:DisLeqk} $|B(v,l)|\in \Theta(l^{d})$.
\item \label{lem:item:DisEqkCover} The nodes in $B(v, l) \backslash B(v,l-l')$ can be covered with $O({(\frac{l}{l'})}^{d-1})$ balls of the form $B(u,l')$ with $u \in B(v,l)\backslash B(v,l-l')$.
\item \label{lem:item:DisLeqkCover} The nodes in $B(v, l)\backslash B(v,l-l'')$ can be covered by $O({(\frac{l}{l'})}^d)$ balls of the form $B(u,l')$ with $u \in B(v,l)\backslash B(v,l-l'')$.
\end{enumerate}
\end{obs}

\subsection{Balance: $\alpha = d+1$}
Here we focus on the case when $\alpha = d+1$ and show that both maximum degree and the iterations required for greedy routing to find its destination are bounded by $O(\log n)$.

Let us start with a useful lemma that provides a lower bound on the drop in the separation-cost when adding a link.
\begin{lemma}
\label{lem:CrstyDrop}
If a node $v$ has no link in $B(u,l)$, then adding a link to $u$ will decrease the total separation-cost of $v$ by $\Omega (l^{d+1})$.
\end{lemma}
\begin{proof}
Consider the set of nodes $S$ with distance $l'$ from $u$ where $l' \leq \frac{1}{3}l$. Note that any node in $S$ induces a separation-cost of at least $\frac{2}{3} l$. If $v$ adds a link to $u$, then the separation-cost of each node in $S$ will be reduced to at most $\frac{1}{3} l$. Therefore, by Item \ref{lem:item:DisLeqk} of Lemma \ref{lem:basiclist} the total decrease in separation-cost is at least $c \frac{l}{3}^d \times \frac{l}{3} = c \frac{l}{3}^{d+1}$ which concludes the lemma.
\end{proof}

Consider an add-stable node $v$, that is, a node $v$ that does not benefit from adding any more links to any other nodes. Next we provide an upper bound on the distance that a node $u$ with $d(u,v)=l$ may have from the neighbors of $v$, that is, $d(u,N(v))$. This is done by showing that if the distance is too large then $v$ would benefit from adding a link to $u$.
\begin{lemma}
\label{lem:TravelBound}
If a node $v$ in $N$ is add-stable, then for any $u$ with $d(v,u)\leq l$ there exists a node $u'$ of distance at most $O(l^{\frac{\alpha}{d+1}})$ from $u$ that $v$ has a link to, that is, $N(v) \cap B(u, c l^\frac{\alpha}{d+1}) \neq \emptyset$ for some constant $c$.
\end{lemma}
\begin{proof}
For the sake of contradiction, suppose we have a node $u$ such that $N(v) \cap B(u, cl^\frac{\alpha}{d+1}) = \emptyset$. By Lemma \ref{lem:CrstyDrop} adding a link from $v$ to $u$ decreases the separation-cost by at least $c' l^\alpha$ for some constant $c'$. When $\beta<c'$, this implies that the reduction in the separation-cost is larger than the cost of a link to $u$ which is at most $\beta l^{\alpha}$, contradicting add-stability.
\end{proof}

Now we can prove an upper bound on the number of iterations required for the greedy routing algorithm to reach its destination.
\begin{theorem}
If a network $N$ is add-stable and $\alpha=d+1$, then any node $u$ is reachable from any node $v$ via greedy routing in $O(\log(n)) $ steps.
\end{theorem}
\begin{proof}
When $\alpha=d+1$, by Lemma \ref{lem:TravelBound} we know that in each iteration of greedy routing the distance to destination is reduced by a constant factor, resulting in logarithmic bound on the number of iterations until the destination is reached.
\end{proof}

Next we turn our focus towards bounding the maximum degree when $\alpha = d+1$ by showing that the number of links that $v$ has to all the nodes of distance between $l$ and $cl$ for some constant $c$, is bounded by a constant.
\begin{lemma}
\label{lem:degree_bound_basic}
When $\alpha \leq d+1$, if a node $v$ in $N$ is toggle-stable then $v$ has at most a constant number of 
links to any ball of radius $l'\leq \epsilon l^{\frac{\alpha}{d+1}}\leq \epsilon l$ centered at a node $u$ with $d(u,v) = l$, that is, $|N(v) \cap B(u, l')| \leq c$.
\end{lemma}
\begin{proof}
Any node within $B(u,l')$ has distance at most $l+l'$ from $v$, and since $v$ in $N$ is toggle-stable, by Lemma \ref{lem:TravelBound}
we can argue that any node within the ball can serve only nodes of distance at most $c'(l+l')^{\frac{\alpha}{d+1}}$ for some constant $c'$.
Therefore any node that is served by a node in $B(u,l')$ must also lie within the ball $B(u,l'+c'(l+l')^{\frac{\alpha}{d+1}})$.
Since $l'\leq \epsilon l^{\frac{\alpha}{d+1}}\leq \epsilon l$, we have
$ l' + c'(l+l')^{\frac{\alpha}{d+1}} \leq c'' l^\frac{\alpha}{d+1}$ for some constant $c''$. Let $C(u')$ be the set of the nodes whose closest node in $N(v)$ is $u'$.
By Item \ref{lem:item:DisLeqk} of Lemma \ref{lem:basiclist} we can argue that
$$\sum_{u'\in N(v) \cap B(u,l')} |C(u')| \leq c^*(l^\frac{\alpha}{d+1})^d.$$
for some constant $c^*$.

By removing a link to a node $u'\in B(u,l')$, the separation-cost induced to $v$ is increased by at most $c'' l^\frac{\alpha}{d+1} |C(u')|$, and therefore there is a node $u' \in  N(v) \cap B(u,l')$ the removal of which decreases the total cost induced to $v$ by at least
$$\beta {(l-l')}^\alpha - c'' l^\frac{\alpha}{d+1} |C(u')| \geq \beta {(l-l')}^\alpha - c'' l^\frac{\alpha}{d+1} \times \frac{c^* l^\frac{\alpha}{d+1})^d}{|N(v) \cap B(u,l')|} \geq \beta(1-\epsilon)^\alpha l^\alpha - \frac{c^\star l^\alpha}{|N(v) \cap B(u,l')|}. $$

Therefore if $|N(v) \cap B(u,l')| > \frac{c^\star}{\beta (1-\epsilon)^\alpha} = c$, $v$ would be able to decrease its total cost by removing $u'$, contradicting stability.
\end{proof}

\begin{theorem}
If a node $v$ in $N$ is toggle-stable and $\alpha=d+1$, then the degree of $v$ is bounded by $O(\log n)$ links.
\end{theorem}
\begin{proof}
For any integer $0\leq i\leq log_{1+\epsilon} (n)$ let $P(i)$ be the set of grid points that are of distance $l$ from $v$ such that $(1+\epsilon)^i <l \leq (1+\epsilon)^{i+1}$. By Lemma \ref{lem:degree_bound_basic} for any $p\in P(i)$ we know that the number of links that $v$ has to the members of the ball $B(p, \epsilon (1+\epsilon)^i)$ is bounded by constant. Also, by Item \ref{lem:item:DisEqkCover} of Lemma \ref{lem:basiclist}  the nodes in $P(i)$ can be covered using $c{\frac{(1+\epsilon)^{i+1}}{\epsilon (1+\epsilon)^i}}^{d-1}$ such balls for some constant $c$. Therefore, for each $0\leq i\leq log_{1+\epsilon} (n)$ we can bound the size of $N(v) \cap P(i)$ by a constant, concluding the theorem.
\end{proof}

\subsection{Sub-logarithmic diameter: $\alpha < d+1$}
In this section we consider cases when $\frac{1}{\log_{\beta^{-1}} n} <\alpha < d+1$ and prove upper bounds on the number of iterations that greedy routing requires and maximum degree for this range of $\alpha$.

\begin{theorem}
If a network $N$ is add-stable and $\alpha<d+1$, then any node $u$ is reachable from any node $v$ via greedy routing in $O(\log_{\frac{d+1}{\alpha}} \log(n)) $ steps.
\end{theorem}
\begin{proof}
By Lemma \ref{lem:TravelBound} if in an iteration of greedy routing the distance to destination is $l$, then in next iteration the distance is at most $cl^\frac{\alpha}{d+1}$ for some constant $c$. Therefore, after
$$\log_{\frac{\alpha}{d+1}} (\frac{1}{\log(n)}) = \log_{\frac{d+1}{\alpha}} \log(n)$$
iterations the distance to destination will be constant.
\end{proof}

\begin{theorem}
If a node $v$ in $N$ is toggle-stable and $\alpha<d+1$, then the degree of $v$ is bounded by $o(\log \log (n) n^{1-\delta\frac{\alpha}{d+1}})$ links for any $\delta<1$.
\end{theorem}
\begin{proof}
For some positive constant $0<a<1$ and some integer $0\leq i\leq \log_{a^{-1}} \log (n)$ let $P(i)$ be the set of grid points that are of distance $l$ from $v$ such that $n^{a^{i+1}/d}<l \leq n^{a^i/d}$. By Lemma \ref{lem:degree_bound_basic} for any $p\in P(i)$ we know that the number of links that $v$ has to the members of the ball $B(p,\epsilon(n^{a^{i+1}/d})^\frac{\alpha}{d+1})$ is bounded by constant.

By Item \ref{lem:item:DisLeqkCover} of Lemma \ref{lem:basiclist} $P(i)$ can be covered with $O(n^{a^i-a^{i+1}\frac{\alpha}{d+1}})$ balls $B(p,\epsilon(n^{a^{i+1}/d})^\frac{\alpha}{d+1})$ where $p\in P(i)$.

Therefore the total number of links of $v$ can be bounded by
$$|N(v)|\leq c + \sum_{i=0}^{i \leq \log_{a^{-1}} \log (n)} n^{a^i-a^{i+1}\frac{\alpha}{d+1}} \leq c + \log_{a^{-1}} \log (n) n^{1-a\frac{\alpha}{d+1}}$$
for some constant $c$.
\end{proof}

\subsection{Sub-logarithmic degrees: $\alpha > d+1$}
Now let us turn our focus to cases when $d+1 < \alpha < \log(\beta^{-1}n^2)$ and prove bounds analogous to those of previous sections. Here we need to modify some of our lemmas, as in this case the cost of links increases relative to length so fast that it no longer benefits $v$ to establish a link to the middle of a region it wishes to cover with that link.

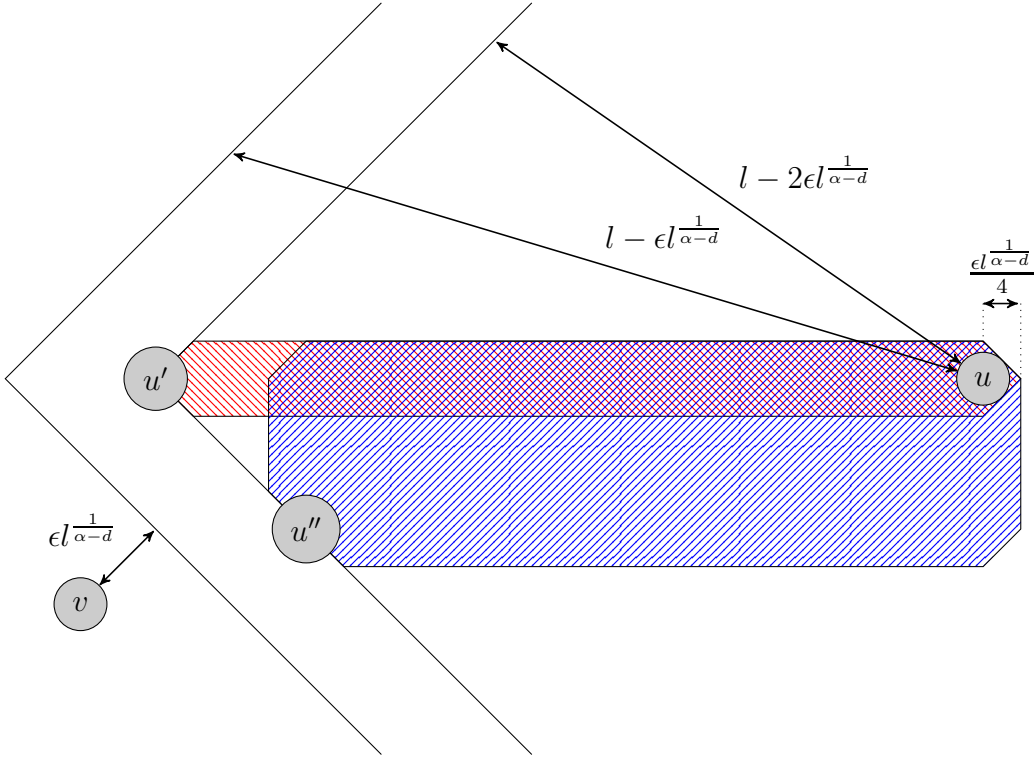
\begin{figure}
\centering
\begin{tikzpicture}[darkstyle/.style={circle,draw,fill=gray!40,minimum size=20}]
\begin{scope}[scale = 1]
		\draw [pattern=north west lines, pattern color=red] (4,6) -- (4.5,6.5) -- (15,6.5) -- (15.5,6) -- (15,5.5)--(4.5,5.5) -- (4,6);
		\draw [pattern=north east lines, pattern color=blue] (6,4) -- (5.5,4.5) -- (5.5,6) -- (6,6.5) -- (15,6.5) -- (15.5,6) -- (15.5,4) -- (15,3.5)-- (6.5,3.5) -- (6,4);
	\draw (7,1) -- (2,6) -- (7,11);
	\draw (9,1) -- (4,6) -- (9,11);

	\node [darkstyle]  (v) at (3,3) {$v$};
	\node [darkstyle]  (u) at (15,6) {$u$};
	\node [darkstyle]  (up) at (4,6) {$u'$};
	\node [darkstyle]  (upp) at (6,4) {$u''$};
	\draw [<->,>=stealth',shorten >=1pt, auto, semithick, scale = 1, transform shape] (u)  -- node[above right] {$l-\epsilon l^\frac{1}{\alpha-d}$}  (5,9);
	\draw [<->,>=stealth',shorten >=1pt, auto, semithick, scale = 1, transform shape] (u)  -- node[above right] {$l-2\epsilon l^\frac{1}{\alpha-d}$}   (8.5,10.5);
	\draw [<->,>=stealth',shorten >=1pt, auto, semithick, scale = 1, transform shape] (v)  -- node[above left] {$\epsilon l^\frac{1}{\alpha-d}$}   (4,4);
	\draw [<->,>=stealth',shorten >=1pt, auto, semithick, scale = 1, transform shape] (15,7)  -- node[above] {$\frac{\epsilon l^\frac{1}{\alpha-d}}{4}$}   (15.5,7);
	\draw [dotted] (15,7) -- (15,6.5);
	\draw [dotted] (15.5,7) -- (15.5,6);
\end{scope}
\end{tikzpicture}
\caption{For Lemma \ref{lem:lalpha:TravelBound}. Adding a link from $v$ to $u'$ (resp. $u''$) decreases the separation-cost of each of the nodes in the area indicated by red (resp. blue) hashing by at least $\frac{\epsilon l^\frac{1}{\alpha-d}}{4}$.}
\label{fig:lem:lalpha:TravelBound}
\end{figure}

\begin{lemma}
\label{lem:lalpha:TravelBound}
If a node $v$ in $N$ is add-stable, then for any $u$ with $d(v,u)\leq l$ there exists a node $u'\in N(v)$ with $d(u,u')\leq l-\epsilon l^\frac{1}{\alpha-d}$.
\end{lemma}
\noindent\emph{Sketch of proof:}
For the sake of contradiction assume that $v$ has no link to any node in $B(u,l-\epsilon l^\frac{1}{\alpha-d})$ and let $u'$ be $\arg \min_{u' \in B(u,l-2\epsilon l^\frac{1}{\alpha-d})} d(u',v)$. Then, the node $v$ benefits from adding a link to $u'$, as $u'$ will be able to reduce the separation-cost of $\Omega((l^\frac{1}{\alpha-d})^{d-1}l)$ nodes by at least $\Omega(l^\frac{1}{\alpha-d})$ (See Figure \ref{fig:lem:lalpha:TravelBound}), and therefore reduce its total separation-cost by $\Omega(l^{1+\frac{d}{\alpha-d}})$. We can argue that
$$\epsilon' l^{1+\frac{d}{\alpha-d}} = \epsilon' l^{\frac{\alpha}{\alpha-d}} > \beta(2\epsilon l^\frac{1}{\alpha-d})^\alpha$$
for proper value of $\epsilon$ and $\beta$, and therefore it is beneficial for $v$ to add a link to $u'$. \qed

\begin{theorem}
If a network $N$ is add-stable and $\alpha>d+1$, then any node $u$ is reachable from any node $v$ via greedy routing in $o(\log\log(n) n^{1-\frac{\delta}{\alpha-d}})$ steps where $\delta$ is any constant smaller than $1$.
\end{theorem}
\begin{proof}
By induction we can argue that starting from a node $v$, the worst case of our analysis of greedy routing is realized when the destination is the farthest node from $v$. By Lemma \ref{lem:lalpha:TravelBound} in each step if the distance to the furthest node is $l$, in the next step this distance will be at most $l-\epsilon l^\frac{1}{\alpha-d}$. Let us consider the number of iterations it takes to get from distance $n^{a^i}$ from destination to distance $l^{a^{i+1}}$ for some $a<1$. This is bounded by at most
$$\frac{l^{a^i}}{\epsilon l^\frac{a^{i+1}}{\alpha-d}}= \epsilon^{-1}l^{a^i-\frac{a^{i+1}}{\alpha-d}}.$$
Therefore, the total number of iterations is bounded by
$$\sum_{i=0}^{i\leq \log\log n} \epsilon^{-1}l^{a^i-\frac{a^{i+1}}{\alpha-d}} \leq \sum_{i=0}^{i\leq \log\log n} \epsilon^{-1}l^{1-\frac{a}{\alpha-d}} \leq \log \log (n) \epsilon^{-1}l^{1=\frac{a}{\alpha-d}}$$
\end{proof}

Here we show that for $d=1$, when $\alpha>d+1$ the degrees are restricted by a sub-logarithmic upper bound. Equivalent proof for higher dimensions remains a challenge.
\begin{theorem}
\label{disoneproof}
If $\alpha>d+1=2$ and a network $N$ is toggle-stable, then any node $v$ has at most $O(\log_{\frac{\alpha}{d+1}} \log(n))$ links.
\end{theorem}
\begin{proof}
Think of each node of distance $l$ from $v$ to initially have a budget of $l$ of separation-cost that can be spent
to justify maintenance of a link. We will show that if $v$ were to have more than $c\log_{\frac{\alpha}{d+1}} \log(n)$ links,
then there would be a set of nodes that would have insufficient budget in order to justify a set of links that serve a subset of them,
and therefore, $N$ would not be toggle-stable.

Consider a node $v$ and let $P(i)$ be the set of nodes with distance $l'$
from $v$ such that $n^\frac{a^{i+1}}{d}< l' \leq n^\frac{a^i}{d}$ where $a = \frac{d+1}{\alpha}$. 
We show that the size of $P(i) \cap N(v)$ is bounded by a constant.
For this, note that any link that $v$ has in $P(i)$ costs at least
$\beta n^\frac{\alpha a^{i+1}}{d}$.

Noting that $d=1$, let $S$ be the set $P(i) \cap N(v)$ minus the elements that are furthest from $v$ on each side, and hence $|S|\geq |P(i)\cap N(v)|-2$.
No node in $S$ can serve a node of distance larger than $n^\frac{a^i}{d}$ from $v$, and by Item \ref{lem:item:DisLeqk} of Lemma \ref{lem:basiclist}
there are at most ${n^\frac{a^i}{d}}^d$ nodes that can contribute to justify $S$,
and each such a node has a budget of at most $n^\frac{a^i}{d}$. Therefore, we have
$$|S| \leq \frac{n^{\frac{a^i}{d}^{d+1}}}{\beta n^\frac{\alpha a^{i+1}}{d}} = \frac{n^{ \frac{(d+1)^{i+1}}{d\alpha^i}}}{\beta n^\frac{ (d+1)^{i+1}}{d\alpha^i}} = \beta^{-1}$$.
\end{proof}
We conjecture that the bound from Theorem \ref{disoneproof} can be extended to $d>1$.

\section{Discussion}
\label{sec:conclusion}
Our main contribution is a deterministic framework for distributed network creation that respects certain theoretical constraints that are demonstrably present in Milgram's experiment.

While the questions that we are concerned with here are inherently inspired by nature,
the results were achieved during a purely mathematical mission to reestablish the logarithmic bounds on maximum degree and routing diameter in a setting where individuals have only local information. However, the model that emerges may give us insight into drives that connect us.

Imagine a population distributed over a plane, where each individual maintains the status (say current address)
of some of the others in order to gain some privilege through them.
There is a real effort here, which is mainly what is required to maintain the status of another,
and there is a sort of anxiety, a pull, an implicit drive that may be described as a desire to have connections everywhere.

Our results suggest that such a population forms a small-world network.


This research also provides us with a pattern for design of robust networks. While they only require limited communication capacities, these networks are highly efficient and simultaneously survivable. It would be interesting to see if in practice they can be used to connect our swarms, e.g., cars on roads.

It must be noted that our cost function can be modified in a number of ways while still yielding networks with small-world properties. For instance, if each node $v$ chooses the set $N(v)$ of those whose status it watches for by minimizing:
\begin{align}
\label{fml:weigthedcost}
c_N(v) = \beta\sum_{u\in N(v)} d^\alpha(v,u) + \sum_{u\notin N(v)}w_v(u)d(u,N(v))
\end{align}
where $w_v(u)$ represents the relative importance of $u$ for $v$, then some degree of small-world properties are established as a side effect depending on parameters $\alpha$ and $\beta$, and the correlation between relative importance and proximity.

At this point, many directions remain unexplored. Particularly, it is curious as to what happens to the threshold for other values of $\beta$. When we set $\beta=n$, our cost function closely resembles the objective function upon which the model of Even-Dar and Kearns \cite{Even-DarK06} is based. Although their model is inherently different, a similar threshold with features of its own appears there. All of this hints towards a bigger picture, perhaps a more general framework that is yet to be discovered.

\section*{Acknowledgement}
I am thankful to Jon Kleinberg for his instructions, and to anonymous reviewers of PODC 16 and SODA 18 for their helpful comments.

\newpage
\bibliographystyle{abbrv}
\bibliography{Social_network_formation}

\end{document}